\documentclass[a4paper,11pt]{llncs}
  \usepackage{enumerate}
  \usepackage{verbatim}
  \usepackage{stfloats}
  \usepackage{graphicx}
  \usepackage{bold-extra}
  \usepackage{amsfonts,amsmath}
  \usepackage{tikz}
  \usetikzlibrary{calc}

  \usepackage[algo2e, noend, noline, boxed]{algorithm2e}
  \SetKwComment{tcm}{\{}{\}}

  \newcommand\twocol[2]{%
  \begin{center}%
  \begin{minipage}[t]{0.55\textwidth}%
  \vspace{0pt}%
  {#1}%
  \end{minipage}\hfill%
  \begin{minipage}[t]{0.45\textwidth}%
  \vspace{0pt}%
  {#2}%
  \end{minipage}%
  \end{center}}

  \newtheorem{observation}{Observation}
  \newtheorem{myclaim}{Claim}
  \def\dotdot{\mathinner{\ldotp\ldotp}}

  \newcommand{\SufCodes}{\mathit{SufCodes}}
  \newcommand{\shape}{\mathit{shape}}
  \newcommand{\code}{\mathit{code}}
  \newcommand{\suf}{\mathit{suf}}
  \newcommand{\prevlt}{\mathit{prev}_<}
  \newcommand{\preveq}{\mathit{prev}_=}
  \newcommand{\opSufTree}{\mathit{opSufTree}}
  \newcommand{\Occ}{\mathit{Occ}}

  \title{Order-Preserving Suffix Trees and\\ Their Algorithmic Applications}

  \author{
    Maxime Crochemore\inst{4,6}
    \and
    Costas S.\ Iliopoulos\inst{4,5}
    \and
    Tomasz Kociumaka\inst{1}
    \and
    Marcin Kubica\inst{1}
    \and
    Alessio Langiu\inst{4}
    \and
    Solon P. Pissis\inst{6,7}\thanks{
    Supported by the NSF--funded iPlant Collaborative (NSF grant \#DBI-0735191).
    }
    \and\\
    Jakub Radoszewski\inst{1}
    \and
    Wojciech Rytter
    \inst{1}\fnmsep\inst{3}\thanks{
    Supported by grant no.\ N206 566740 of the National Science Centre.
    }
    \and
    Tomasz Wale\'n\inst{2,1}
  }
  \institute{
    Faculty~of Mathematics, Informatics and Mechanics,\\
    University of Warsaw, Warsaw, Poland\\
    \email{[kociumaka,jrad,rytter,walen]@mimuw.edu.pl}
    \and
    Laboratory of Bioinformatics and Protein Engineering,\\
    International Institute of Molecular and Cell Biology in Warsaw, Poland
    \and
    Faculty of Mathematics and Computer Science,\\
    Copernicus University, Toru\'n, Poland
    \and
    Dept.~of Informatics, King's College London, London WC2R 2LS, UK \\
    \email{[maxime.crochemore,csi]@dcs.kcl.ac.uk}
    \and
    Faculty of Engineering, Computing and Mathematics,\\
    University of Western Australia, Perth WA 6009, Australia
    \and
    Universit\'e Paris-Est, France
    \and
    Laboratory of Molecular Systematics and Evolutionary Genetics,\\
    Florida Museum of Natural History, University of Florida, USA
    \and
    Scientific Computing Group (Exelixis Lab \& HPC Infrastructure),\\
    Heidelberg Institute for Theoretical Studies (HITS gGmbH), Germany\\
    \email{solon.pissis@h-its.org}
  }

\begin{document}
  \maketitle

  \begin{abstract}
    Recently Kubica et al.\ (\emph{Inf.\ Process.\ Let.}, 2013) and
    Kim et al.\ (\emph{submitted to Theor.\ Comp.\ Sci.}) introduced
    order-preserving pattern matching.
    In this problem we are looking for consecutive substrings of the text
    that have the same ``shape'' as a given pattern.
    These results include a linear-time order-preserving pattern
    matching algorithm for polynomially-bounded alphabet and
    an extension of this result to pattern matching with multiple patterns.
    We make one step forward in the analysis and give an
    $O(\frac{n\log{n}}{\log\log{n}})$ time randomized algorithm
    constructing suffix trees in the order-preserving setting.
    We show a number of applications of order-preserving suffix trees
    to identify patterns and repetitions in time series.
  \end{abstract}

  \section{Introduction}
  We introduce order-preserving suffix trees that can be used for pattern matching
  and repetition discovery problems in the order-preserving setting,
  in particular, to model finding trends in time series
  which appear naturally when considering e.g. the stock market or
  melody matching of two musical scores.

  Two strings $x,\ y$ of the same length over an integer alphabet
  are called \emph{order-isomorphic} (or simply isomorphic), written $x\approx y$, if
  $$\forall_{1\le i,j \le |x|}\ x[i] \le x[j]\Leftrightarrow y[i] \le y[j].$$
  \begin{example}
    $(5, 2, 7, 5, 1, 4, 9, 4, 5) \approx (6, 4, 7, 6, 3, 5, 8, 5, 6)$, see Fig.~\ref{fig:matching}.
  \end{example}
  The notion of order-isomorphism was introduced in \cite{Costas_TCS} and \cite{Kulczynski_IPL}.
  Both papers independently study the problem of identifying all consecutive substrings
  of a string $x$ that are order-isomorphic to a given string $y$, the so-called
  order-preserving pattern matching problem.
  If $|x|=n$ and $|y|=m$, an $O(n+m\log{m})$ time algorithm for this problem
  is presented in both papers.
  Morover, \cite{Costas_TCS} presents extensions of this problem to multiple-pattern matching
  based on the algorithm of Aho and Corasick.

  The problem of order-preserving pattern matching has evolved from the combinatorial
  study of patterns in permutations.
  This field of study is concentrated on pattern avoidance, that is, counting the number of
  permutations not containing a subsequence which is order-isomorphic to a given pattern.
  Note that in this problem the subsequences need not to be consecutive.
  The first results on this topic were given by
  Knuth \cite{Knuth_Art1} (avoidance of 312),
  Lov{\'a}sz \cite{lovasz} (avoidance of 213)
  and Rotem \cite{Rotem} (avoidance of both 231 and 312).
  On the algorithmic side, patten matching in permutations (as a subsequence)
  was shown to be NP-complete \cite{Bose_npcomplete}
  and a number of polynomial-time algorithms for special cases of patterns
  were developed \cite{Albert,ChangWang,DBLP:conf/isaac/GuillemotV09,LIbarra}.

  \paragraph{\bf Structure of the paper.}
  In Section~\ref{sec:op_suftree} we give a formal definition of an order-preserving
  suffix tree and describe its basic properties.

  To obtain an efficient algorithm constructing such suffix trees in
  Section~\ref{sec:offline_oracle} we develop
  an offline character oracle based on orthogonal range counting.
  An $O(\frac{n\log{n}}{\log\log{n}})$ time randomized and offline algorithm
  constructing order-preserving suffix trees is obtained.
  It is based on a general framework of
  Cole and Hariharan \cite{DBLP:journals/siamcomp/ColeH03a}
  (or, alternatively, on the approach of Lee, Na and Park \cite{Lee2011201}).

  Finally in Section~\ref{sec:applications} present a number of applications of order-preserving suffix trees
  that generalize the results from \cite{Costas_TCS,Kulczynski_IPL}.
  These applications are based on classical applications of suffix trees,
  however a new combinatorial insight is required to adapt the known tools
  to the order-preserving setting.
  In particular, we consider order-preserving string matching
  and the problem of detecting the simplest
  order-preserving repetitions that we call op-squares.

  \section{Preliminaries}
  Let $w$ be a string of length $n$ over an integer alphabet $\Sigma$, $w=w_1 \ldots w_n$.
  We assume that $\Sigma$ is polynomially bounded in terms of $n$.
  By $w[i \dotdot j]$ we denote the substring $w_i \dotdot w_j$.
  Denote by $\suf_i$ the $i$-th suffix of $w$, that is, $w[i \dotdot n]$.
  For any $i \in \{1,\ldots,n\}$ define:
  \begin{align*}
    \prevlt(w,i) &= |\{k\,:\, k < i,\,w_k < w_i\}|,\\ 
    \preveq(w,i) &= |\{k\,:\, k < i,\,w_k = w_i\}|.
  \end{align*}
  We introduce codes of single positions and and codes of strings as follows:
  $$\phi(w,i)=(\prevlt(w,i),\preveq(w,i))$$
    $$\code(w)=(\phi(w,1),\phi(w,2),\ldots,\phi(w,n)).$$
 \noindent
  For a string $w$, define $\shape(w)$ as the lexicographically smallest
  string $u$ over $\{0,1,\ldots\}$ such that $u \approx w$.

  \begin{figure}[ht]
  \centering
    \includegraphics{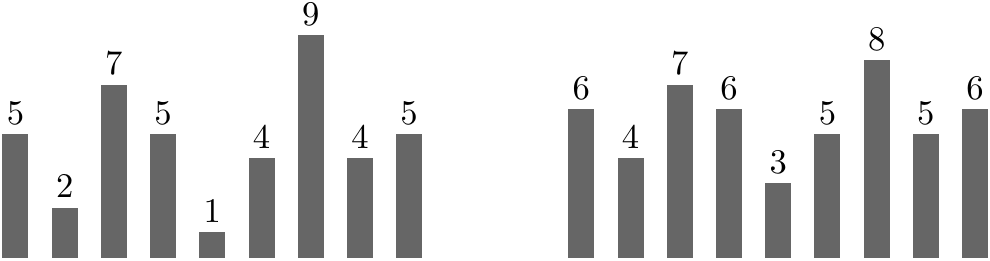}
    \caption{
      Example of two order-isomorphic strings. Their codes are equal to $(0,0) \,    \, (0,0) \,    \, (2,0) \,    \, (1,1) \,    \, (0,0) \,
           \, (2,0) \,    \, (6,0) \,    \, (2,1) \,    \, (4,2)$ and their shapes are equal to $(3, 1, 4, 3, 0, 2, 5, 3, 3)$.
    }\label{fig:matching}
  \end{figure}

  \begin{observation}\label{obs:prefix_properties}
    The code has an {\em online property}: the code of the $i$-th character does
    not depend on characters in positions to the right of $i$:
    if $\code(x)=\code(y)$ then $\code(x)$ is a prefix of $\code(yz)$.
    (Note that the function $\shape$ does not have this property.)
  \end{observation}

  \noindent
  The following obvious fact is useful in the proof of the forthcoming lemma.
  \begin{observation}
    $x \approx y\ \Leftrightarrow\ \shape(x)=\shape(y)$.
  \end{observation}

  \begin{lemma}\label{lem:code}
    $x \approx y\ \Leftrightarrow\ \code(x)=\code(y)$.
  \end{lemma}
  \begin{proof}
    The $(\Rightarrow)$ part of the equivalence follows from the definition of a code.
    As for the $(\Leftarrow)$ part, we show an algorithm that reconstructs $\shape(x)$ from $\code(x)$.
    Thus $\code(x)=\code(y)$ implies that $\shape(x)=\shape(y)$ which, in turn, implies
    $x \approx y$.

    The algorithm is as follows.
    Find the rightmost $(0,0)$ in $\code(x)$ (it exists, since
    $\code(x)$ starts with a $(0,0)$).
    Find all elements of the form $(0,z)$ to the right of this $(0,0)$.
    All these elements together with this $(0,0)$ are equal and they
    correspond to 0s in $\shape(x)$.
    Remove all these elements and decrease the first coordinate
    of every other element in $\code(x)$ by the number of removed elements
    that were to its right and repeat the process from the beginning,
    identifying all 1s, 2s etc in $\shape(x)$.
  \qed
  \end{proof}

  \section{Order-Preserving Suffix Trees}\label{sec:op_suftree}
  Let us define the following family of strings:
  $$\SufCodes(w)=\{\code(\suf_1)\#,\,\code(\suf_2)\#,\,\ldots,\,\code(\suf_n)\#\},$$
  see Fig.~\ref{fig:SufCodes}.
  The \emph{order-preserving suffix tree} of $w$, denoted $\opSufTree(w)$,
  is a compacted trie of all the sequences in $\SufCodes(w)$.
  The $\opSufTree(w)$ contains $O(n)$ leaves, hence its size is $O(n)$.

\begin{figure}
{\footnotesize
\begin{verbatim}
          suffixes of w:                       SufCodes(w):

        6 8 2 0 7 9 3 1 4 5                0 1 0 0 3 5 2 1 4 5 #
          8 2 0 7 9 3 1 4 5                  0 0 0 2 4 2 1 4 5 #
            2 0 7 9 3 1 4 5                    0 0 2 3 2 1 4 5 #
              0 7 9 3 1 4 5                      0 1 2 1 1 3 4 #
                7 9 3 1 4 5                        0 1 0 0 2 3 #
                  9 3 1 4 5                          0 0 0 2 3 #
                    3 1 4 5                            0 0 2 3 #
                      1 4 5                              0 1 2 #
                        4 5                                0 1 #
                          5                                  0 #
\end{verbatim}
}
\caption{
  $\SufCodes(w)$ for $w=(6,8,2,0,7,9,3,1,4,5)$.
  In this example all the characters of the string $w$ are distinct, hence for
  each $i$ we have $prev_{=}(w,i)=0$ and we can ignore the second components of $\phi$.
  It suffices to take the first component of the code.
  \label{fig:SufCodes}
}
\end{figure}

  As usual, only the explicit nodes (that is, branching nodes and leaves)
  of $\opSufTree(w)$ are stored.
  The leaves store starting positions of the corresponding suffixes.
  Each branching node stores its depth and one of the leaves in its subtree.
  Each inner node stores a suffix link that may lead to an implicit or an explicit node.

  Each edge stores the code only of its first character.
  The codes of all the remaining characters of any edge can be obtained
  using the so-called \emph{character oracle} that can efficiently
  provide the code $\phi(\suf_i,j)$ for any $i$ and $j$
  (a decription of the character oracle construction is given in Section~\ref{sec:offline_oracle}).

  \begin{example}
    Consider the order-preserving suffix tree of the string
    $$w=(6,8,2,0,7,9,3,1,4,5),$$
    see Fig.~\ref{fig:op_suf_tree}.
    All $\SufCodes(w)$ are given in Fig.~\ref{fig:SufCodes}.

    \begin{figure}[htpb]
      \begin{center}
        \twocol{
          \includegraphics[width=6cm]{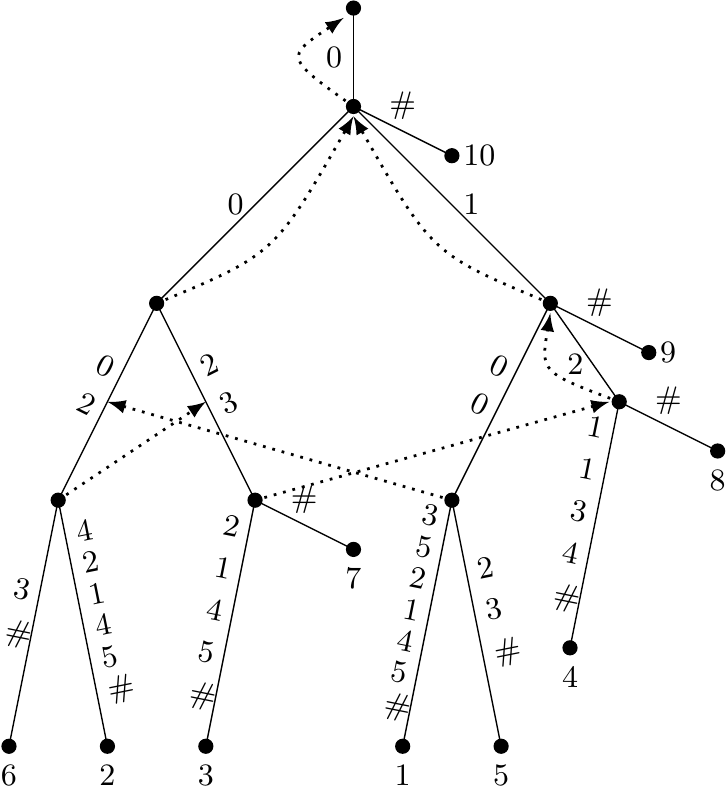}
        }{
          \includegraphics[width=5.5cm]{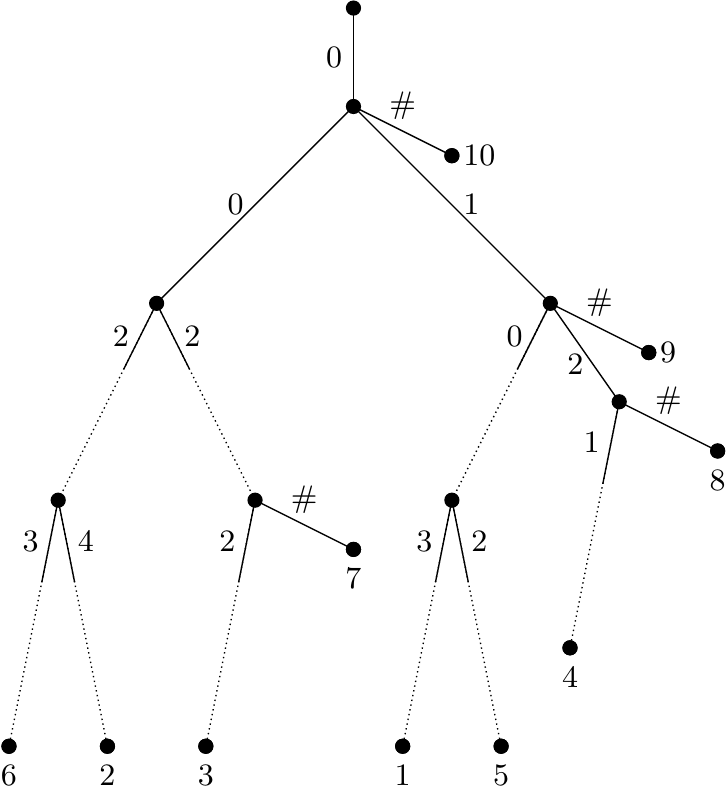}
        }
        \caption{\label{fig:op_suf_tree}
          The uncompacted trie of $\SufCodes(w)$ for $w=(6,8,2,0,7,9,3,1,4,5)$ (to the left)
          and its compacted version which together with character oracle forms
          $\opSufTree(w)$ (to the right).
        }
      \end{center}
    \end{figure}
  \end{example}

  \section{Character Oracle}\label{sec:offline_oracle}
  We use a geometric approach: the computation of $\phi$ for $w$ corresponds
  to counting points in certain orthogonal rectangles in the plane.

  \begin{observation}\label{obs:points_plane}
    Let us treat the pairs $(i,w_i)$ as points in the plane.
    Then $\phi(\suf_j,i)=(a,b)$, where $a$ is the number of points
    that lie within the rectangle $A = [j,i-1] \times (-\infty,w_i)$
    and $b$ is the number of points in the rectangle $B = [j,i-1] \times [w_i,w_i]$,
    see Fig.~\ref{fig:orthogonal}.
  \end{observation}

  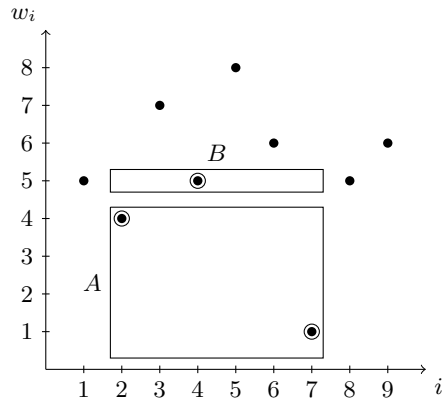
\begin{figure}[htpb]
    \begin{center}

      \begin{tikzpicture}
          \def\unitX{0.5cm}
          \def\unitY{0.5cm}
          \def\queryI{2}
          \def\queryJ{8}
          \def\queryV{5}

          \draw [->] (0,0)--(10*\unitX,0) node [below right] {$i$};
          \draw [->] (0,0)--(0,9*\unitY) node [above left] {$w_i$};

          \foreach \i in {1,2,3,4,5,6,7,8,9} {
              \draw (\i*\unitX,0.1*\unitX)--(\i*\unitX,-0.1*\unitX) node [below] {$\i$};
          }
          \foreach \i in {1,2,3,4,5,6,7,8} {
              \draw (0.1*\unitX,\i*\unitY)--(-0.1*\unitX,\i*\unitY) node [left] {$\i$};
          }
          \foreach \i/\y in {1/5, 2/4, 3/7, 4/5, 5/8, 6/6, 7/1, 8/5, 9/6} {
              \node [fill, circle, scale=0.4] at (\i*\unitX,\y*\unitY) {};
          }

          \draw (\queryI*\unitX-\unitX*0.3, \queryV*\unitY+\unitY*0.3)--
                (\queryJ*\unitX-\unitX*0.7, \queryV*\unitY+\unitY*0.3) node [midway, above] {$B$} --
                (\queryJ*\unitX-\unitX*0.7, \queryV*\unitY-\unitY*0.3)--
                (\queryI*\unitX-\unitX*0.3, \queryV*\unitY-\unitY*0.3)--cycle;

          \draw (\queryI*\unitX-\unitX*0.3, \queryV*\unitY-\unitY*0.7)--
                (\queryJ*\unitX-\unitX*0.7, \queryV*\unitY-\unitY*0.7)--
                (\queryJ*\unitX-\unitX*0.7, \unitY*0.3)--
                (\queryI*\unitX-\unitX*0.3, \unitY*0.3)--
                (\queryI*\unitX-\unitX*0.3, \queryV*\unitY-\unitY*0.7) node [midway, left] {$A$};

          \foreach \i/\y in {2/4, 4/5, 7/1} {
              \draw (\i*\unitX,\y*\unitY) node [circle,minimum size=0.4*\unitX,draw,inner sep=0pt] {};
          }
      \end{tikzpicture}

      \caption{\label{fig:orthogonal}
        Geometric illustration of the sequence $w=(5,4,7,5,8,6,1,5,6)$.
        The elements $w_i$ are represented as points $(i,w_i)$.
        The computation of $\phi(\suf_2,8)=(2,1)$ corresponds to
        counting points in rectangles $A$, $B$.
      }
    \end{center}
  \end{figure}

  \noindent
  The orthogonal range counting problem is defined as follows.
  We are given $n$ points in a plane and we need to answer queries
  of the form:
  
  ``how many points are contained in a given axis-aligned rectangle?''.

  \noindent
  An efficient solution to this problem was given by Chan and
  P\v{a}tra\c{s}cu, see Theorem 2.3\ in~\cite{DBLP:conf/soda/ChanP10}
  which we state below as Lemma~\ref{lem:orthogonal}.
  We say that a point $(p,q)$ dominates a point $(p',q')$ if $p>p'$ and $q>q'$.

  \begin{lemma}\label{lem:orthogonal}
    We can preprocess $n$ points in the plane
    in $O(n \sqrt{\log n})$ time, using a data structure with $O(n)$
    words of space, so that we can count the number of
    points dominated by a query point in $O(\log n/ \log \log n)$ time.
  \end{lemma}

  \noindent
  One can easily observe that the offline orthogonal range counting
  can be reduced to the dominance problem described in Lemma~\ref{lem:orthogonal}.
  We use the solution from this lemma to build our character oracle.

  \begin{lemma}\label{lem:oracle}
    Let $w$ be a string of length $n$ and let $\suf_1,\ldots,\suf_n$ be its suffixes.
    After $O(n\sqrt{\log{n}})$ time and $O(n)$ space preprocessing
    one can compute $\phi(\suf_j,i)$ for any $i$, $j$
    in $O(\log{n}/\log\log{n})$ time.
  \end{lemma}
  \begin{proof}
    Due to Observation~\ref{obs:points_plane} our problem can be reduced to an orthogonal range counting problem.
    Using Lemma~\ref{lem:orthogonal} we obtain a solution to this problem
    with the requested preprocessing and query time and space.
  \qed
  \end{proof}
  

  \section{Construction of Order-Preserving Suffix Trees}\label{sec:op_suftree_construction}
  We use the tools introduced by Cole and Hariharan \cite{DBLP:journals/siamcomp/ColeH03a}
  for the construction of suffix trees for quasi-suffix collections of strings.

  \subsection{Quasi-suffix Collections}

  A family of strings $S_1,\ldots,S_n$ is called
  a quasi-suffix collection \cite{DBLP:journals/siamcomp/ColeH03a}
  if the following conditions hold:
  \begin{enumerate}
    \item $|S_1|=n$ and $|S_i|=|S_{i-1}|-1$ for all $i>1$.
    \item No $S_i$ is a prefix of another $S_j$.
    \item If $S_i$ and $S_j$ have a common prefix of length $l>0$
    then $S_{i+1}$ and $S_{j+1}$ have a common prefix of length
    at least $l-1$.
  \end{enumerate}

  \noindent
  The \emph{suffix tree} for a quasi-suffix collection is defined as a compacted trie
  of all the strings in the collection.

  \begin{lemma}\label{lem:qs_collection}
    Let $w$ be a string of length $n$.
    Then the strings in $\SufCodes(w)$ form a quasi-suffix collection.
  \end{lemma}
  \begin{proof}
    The conditions 1 and 2 of a quasi-suffix collection obviously hold.
    The condition 3 is a direct consequence of the common prefix property
    (cf. \cite{DBLP:journals/jcss/Baker96}):

    \begin{myclaim}
      If $\code(ax)=\code(by)$ then $\code(x)=\code(y)$.
    \end{myclaim}
    \begin{proof}
      Due to Lemma~\ref{lem:code}, $\code(ax)=\code(by)$ implies that $ax \approx by$.
      Hence, obviously $x \approx y$.
      Again due to Lemma~\ref{lem:code} we have $\code(x)=\code(y)$.
      \qed
    \end{proof}

    \noindent
    Consequently, $\SufCodes(w)$ satisfies all conditions for a quasi-suffix collection.
  \qed
  \end{proof}

  \subsection{Order-Preserving Suffix-Tree Construction}\label{sec:nondeterministic_offline}
  Cole and Hariharan \cite{DBLP:journals/siamcomp/ColeH03a} provided
  a general framework for constructing suffix trees for
  quasi-suffix collections $(S_i)$.
  Assuming they are given a character oracle that provides the $j$-th character
  of any $S_i$ in $O(1)$ time, Cole and Hariharan \cite{DBLP:journals/siamcomp/ColeH03a}
  can construct the suffix tree for a quasi-suffix collection in
  $O(n)$ time and space with almost inverse exponential failure probability.
  This result assumes that $S_i$ are over an alphabet of size polynomial in $n$.

  We apply this result to obtain an order-preserving suffix tree
  by using the character oracle that we developed in Section~\ref{sec:offline_oracle}.

  \begin{theorem}
    The order-preserving suffix tree of a string of length $n$
    can be constructed in $O(\frac{n\log{n}}{\log\log{n}})$ randomized time.
  \end{theorem}
  \begin{proof}
    Due to Lemma~\ref{lem:qs_collection}, we can apply the algorithm
    of Cole and Hariharan \cite{DBLP:journals/siamcomp/ColeH03a}.
    We use the character oracle from Lemma~\ref{lem:oracle}.
    Cole and Hariharan \cite{DBLP:journals/siamcomp/ColeH03a} call the oracle
    $O(n)$ time, which gives $O(\frac{n\log{n}}{\log\log{n}})$ total construction time.
  \qed
  \end{proof}

  The framework of Cole and Hariharan \cite{DBLP:journals/siamcomp/ColeH03a}
  is based on McCreight's algorithm for suffix tree construction \cite{DBLP:journals/jacm/McCreight76}
  which is an offline algorithm.
  Recently Lee, Na and Park \cite{Lee2011201} presented a modified
  version of the algorithm from \cite{DBLP:journals/siamcomp/ColeH03a}
  that uses Ukkonen's suffix tree construction algorithm \cite{DBLP:journals/algorithmica/Ukkonen95}
  which works online (the characters of the string can be given one at a time).
  The construction of Lee, Na and Park is designed only for parameterized
  suffix trees, however it works also in the general quasi-suffix setting,
  hence, in particular, in the order-preserving setting.
  Using this construction, an alternative $O(\frac{n\log{n}}{\log\log{n}})$ time
  algorithm for order-preserving suffix trees can be obtained.
  Unfortunately, our oracle does not work online and thus the resulting
  algorithm is still offline.

%

  \section{Applications of Order-Preserving Suffix Trees}\label{sec:applications}
  The most common application of suffix trees is pattern matching with
  time complexity independent of the length of the text.
  With the aid of order-preserving suffix trees we obtain a similar result
  with an additional small factor in the time complexity
  which is due to the character oracle.
  This result is possible due to the ``suffix-independence''
  property of our coding function, see Observation~\ref{obs:prefix_properties}.

  \begin{theorem}\label{thm:pattern_matching}
    Assume that we have $\opSufTree(w)$ of a string $w$ of length $n$.
    Given a pattern $x$ of length $m$, one can check if $x$ is
    a factor of $w$ in $O(\frac{m\log{n}}{\log\log{n}})$ time
    and report all occurrences in $O(\frac{m\log{n}}{\log\log{n}}+\Occ)$ time,
    where $\Occ$ is the number of occurrences.
  \end{theorem}
  \begin{proof}
    First we build the character oracle for the pattern, this takes
    $O(m\sqrt{\log{m}}) = O(\frac{m\log{n}}{\log\log{n}})$ time.
    To answer a query, we traverse down the edges of the suffix tree
    using the character oracles for the pattern and the text.
    If we are at a branching node of depth $h$, we check if there is an outgoing edge
    starting with $\phi(x[1\dotdot h],x[h])$.
    Otherwise we are at an implicit node of depth $h$ located on an edge leading
    to an explicit node that has some leaf $i$ in its subtree.
    In this case we check if
    $\phi(x[1\dotdot k],x[k])$ equals $\phi(w[i\dotdot i+h],w[i+h])$.

    This enables to find the locus of $x$ in $O(\frac{m\log{n}}{\log\log{n}})$ time.
    Afterwards all the occurrences of $x$ can be found in the usual way
    by traversing all nodes in the corresponding subtree.
  \qed
  \end{proof}

%

  \noindent
  A string $uv$ is called an \emph{order-preserving square} (an \emph{op-square}, in short)
  if $u \approx v$.
  The length of the op-square is defined as $|uv|$.
  Thus an op-square represents a repeating pattern in a time series.
  Using order-preserving suffix trees we can obtain algorithms for finding
  and reporting op-squares.

  Note that each string of length at least 2 contains an op-square of length 2.
  Hence, no such string is op-square-free.
  We show how to modify the square-detecting algorithm by
  Gusfield and Stoye \cite{DBLP:journals/tcs/StoyeG02} to check, for each length $k$,
  if a given string $w$ contains an op-square of length $k$.

  \paragraph{\bf Branching squares.}
  We say that a substring $w[i \dotdot i+2k-1]$ is a \emph{branching square}
  if $w[i \dotdot i+k-1] = w[i+k\dotdot i+2k-1]$ and $w[i+2k] \ne w[i]$.
  The algorithm from \cite{DBLP:journals/tcs/StoyeG02} uses the suffix tree of a text $w$,
  $|w|=n$, to find all \emph{branching squares} in $w$ in $O(n\log{n})$ time.
  Each such square is detected when inspecting the edges outgoing from
  the explicit node corresponding to the string $w[i \dotdot i+k-1]$.

  \paragraph{\bf Non-extendible and non-shiftable op-squares.}
  We say that an op-square $w[i \dotdot i+2k-1]$ is \emph{non-extendible}
  if $$w[i \dotdot i+k-1]\approx w[i+k\dotdot i+2k-1]$$
  and
  $$w[i \dotdot i+k]\not\approx w[i+k\dotdot i+2k].$$
  A \emph{non-shiftable} op-square is defined similarly but with
  the last condition substituted by
  $$w[i+1 \dotdot i+k]\not\approx w[i+k+1\dotdot i+2k].$$
  When we apply algorithm from \cite{DBLP:journals/tcs/StoyeG02}
  to the order-preserving suffix tree, we find all non-extendible op-squares.
  It suffices to prove the following property.

  \begin{lemma}
    If $w$ contains an op-square of a given length then it contains a non-extendible
    op-square of the same length.
  \end{lemma}
  \begin{proof}
    Let $w[i \dotdot i+2k-1]$ be the rightmost op-square of length $2k$ in $w$.
    Then it is a non-shiftable op-square:
    $$w[i+1 \dotdot i+k]\not\approx w[i+k+1\dotdot i+2k].$$
    Hence,
    $$w[i \dotdot i+k]\not\approx w[i+k\dotdot i+2k]$$
    and consequently $w[i \dotdot i+2k-1]$ is a non-extendible op-square.
    \qed
  \end{proof}

  \noindent
  Consequently we obtain an efficient algorithm for detecting an op-square of a given length.
  Note that the algorithm does not require to query the character oracle, it only processes the skeleton
  of the suffix tree.

  \begin{theorem}
    For a string $w$ of length $n$, after $O(n\log{n})$ time preprocessing one can
    check if $w$ contains an op-square of a given length in $O(1)$ time.
  \end{theorem}

  The algorithm of Gusfield and Stoye can also compute all the occurrences
  of \emph{squares} in a string in additional time proportional to the number of reported
  occurrences.
  For this, it starts at every branching square $w[i \dotdot i+2k-1]$ and shifts
  it to the left position-by-position as long as it forms a square, i.e.
  as long as $w[i-j]=w[i+k-j]$, $j=1,2,\ldots$

  A generalization of this algorithm to op-squares requires efficient testing
  if an op-square can be shifted to the left.
  This could be done using the character oracle for the reversed text, however, there is
  a more efficient solution.

  \begin{theorem}
    All order-preserving squares in a string $w$ of length $n$ can be computed
    in $O(n\log{n} + \Occ)$ time, where $\Occ$ is the total number of occurrences
    of op-squares.
  \end{theorem}
  \begin{proof}
    We use the following fact.
    \begin{myclaim}\label{claim:LCA}
      The string $w[i \dotdot i+2k-1]$ is an op-square if and only if
      the lowest common ancestor (LCA) node of the leaves
      of $\opSufTree(w)$ corresponding to $\suf_i$ and $\suf_{i+k}$
      has depth at least $k$.
    \end{myclaim}

    After $O(n)$ preprocessing time, LCA of nodes in a tree can be computed
    in $O(1)$ time \cite{DBLP:journals/siamcomp/HarelT84}.
    By the claim we can keep shifting the non-extendible op-square
    to the left.
    We stop either when the tested substring is not an op-square or when
    we encounter another non-extendible op-square, the latter situation is possible since
    non-extendible op-squares can still be shiftable.
    We obtain an algorithm with required complexity.
  \qed
  \end{proof}

  \bibliographystyle{abbrv}
  \bibliography{op_suftree}

\end{document}